\RequirePackage[l2tabu,orthodox]{nag}% nag when using outdated macros
\documentclass[a4paper,english]{scrartcl}
\pdfoutput=1 % make pdf file directly (recommended)

\usepackage{todonotes,hyperref}

\usepackage{array}
\usepackage{multirow}

\usepackage[utf8]{inputenc}
\usepackage[T1]{fontenc}
\usepackage{lmodern}    % Up-to-date computer modern fonts

\usepackage[%
  activate={true,nocompatibility},  % activate protrusion and expansion
  final,                % enable microtype; use "draft" to disable
  tracking=true,        %
  kerning=true,         %
  spacing=true,         %
  ]{microtype}
\microtypecontext{spacing=nonfrench}

\usepackage[within=none,figurename=Fig.,labelsep=period]{caption} % Figure captions

\usepackage{amssymb}
\usepackage{amsmath}
\usepackage{bm}         % bold in math mode: \bm{ }
\usepackage{mathtools}  % for parentheses

\allowdisplaybreaks[1] % allow page breaks in {align} environments

\usepackage{tikz}
\usetikzlibrary{shapes}

\usepackage[inline]{enumitem}% better enumerate environment
\setlist[itemize]{label=$\circ$}
\setlist[description]{labelindent=\parindent}

\setdescription{leftmargin=30pt}

\usepackage[amsmath,hyperref,thmmarks]{ntheorem}

\theoremseparator{.}
\newtheorem{theorem}{Theorem}
\newtheorem{lemma}[theorem]{Lemma}

\newtheorem{observation}[theorem]{Observation}
\newtheorem{corollary}[theorem]{Corollary}

\theoremstyle{plain}
\theorembodyfont{\upshape}

\theoremstyle{nonumberplain}
\theoremheaderfont{\itshape}
\theorembodyfont{\upshape}
\theoremsymbol{\ensuremath{_\blacksquare}}
\newtheorem{proof}{Proof}

\DeclarePairedDelimiter\paren{\lparen}{\rparen}
\DeclarePairedDelimiter\abs{\lvert}{\rvert}
\DeclarePairedDelimiter\set{\{}{\}}
\DeclarePairedDelimiterX\setc[2]{\{}{\}}{\,#1 \;\colon\; #2\,}
\DeclarePairedDelimiterX\parenc[2]{\lparen}{\rparen}{\,#1 \;\delimsize\vert\; #2\,}

\newcommand{\cc}[1]{\ensuremath{\mathrm{#1}}}
\newcommand{\pp}[1]{\textup{#1}}
\newcommand{\op}[1]{\ensuremath{\operatorname{#1}}}

\renewcommand{\P}{\cc{P}}

\newcommand{\poly}{\op{poly}}

\newcommand{\N}{\mathbf{N}}

\newcommand{\Q}{\mathbf{Q}}
\newcommand{\R}{\mathbf{R}}

\newcommand{\GF}[1]{\ensuremath{{\textrm{GF}}(#1)}}

\newcommand{\dotcup}{\mathbin{\dot\cup}}

\newcommand{\zo}{\set{0,1}}

\usepackage{mathrsfs}

\usepackage{xparse}
\DeclareDocumentCommand{\restrict}{O{}}{\mathord{\restriction}_{#1}}
\newcommand{\cETH}{\ensuremath{{\#\mathrm{ETH}}}}
\newcommand{\forests}{\ensuremath{\mathcal F}}

\newcommand{\psat}{\#\mathrm{Pos2Sat}}
\newcommand{\isat}{\#\mathrm{Imp2Sat}}

\title{Fine-grained dichotomies for the Tutte plane and Boolean \#CSP
 \footnote{This work was done while the authors were visiting the Simons Institute for the Theory of Computing.}}

\author{Cornelius Brand, Holger Dell, Marc Roth \\
Saarland University and Cluster of Excellence, MMCI}
\date{\today}

\begin{document}
 
\maketitle

\begin{abstract}
  Jaeger, Vertigan, and Welsh~\cite{JVW90} proved a dichotomy for the complexity 
  of evaluating the Tutte polynomial at fixed points: The evaluation is \#P-hard 
  almost everywhere, and the remaining points admit polynomial-time algorithms.  
  Dell, Husfeldt, and Wahlén~\cite{DHW10} and Husfeldt and 
  Taslaman~\cite{HusfeldtTaslaman}, in combination with 
  Curticapean~\cite{curticapean2015block}, extended the \#P-hardness results to 
  tight lower bounds under the counting exponential time hypothesis \#ETH, with 
  the exception of the line $y=1$, which was left open. We complete the 
  dichotomy theorem for the Tutte polynomial under \#ETH by proving that the 
  number of all acyclic subgraphs of a given $n$-vertex graph cannot be 
  determined in time $\exp\paren[\big]{o(n)}$ unless \#ETH fails.

  Another dichotomy theorem we strengthen is the one of Creignou and Hermann 
  \cite{creignou1996complexity} for counting the number of satisfying 
  assignments to a constraint satisfaction problem instance over the Boolean 
  domain. We prove that all \#P-hard cases are also hard under \#ETH. The main 
  ingredient is to prove that the number of independent sets in bipartite graphs 
  with $n$ vertices cannot be computed in time $\exp\paren[\big]{o(n)}$ unless 
  \#ETH fails.

  In order to prove our results, we use the block interpolation idea by 
  Curticapean~\cite{curticapean2015block} and transfer it to systems of linear 
  equations that might not directly correspond to interpolation.
\end{abstract}

\section{Introduction}

Counting combinatorial objects is at least as hard as detecting their existence, 
and often it is harder.
Valiant~\cite{V79} introduced the complexity class \cc{\#P} to study the 
complexity of counting problems and proved that counting the number of perfect 
matchings in a given graph is \cc{\#P}-complete.
By a theorem of Toda~\cite{toda89}, we know that 
$\cc{PH}\subseteq\cc{P}^\cc{\#P}$ holds; in particular, for every problem in the 
entire polynomial-time hierarchy, there is a polynomial-time algorithm that is 
given access to an oracle for counting perfect matchings.
This theorem suggests that counting is much harder than decision.

When faced with a problem that is \cc{NP}-hard or \cc{\#P}-hard, the area of 
exact algorithms strives to find the fastest exponential-time algorithm for a 
problem, or find reasons why faster algorithms might not exist.
For example, the fastest known algorithm for computing perfect matchings in 
$n$-vertex graphs~\cite{bjorklund2012counting} runs in time $2^{n/2} \cdot 
\poly(n)$.
It has been hypothesized that no $O\paren*{1.99^{n/2}}$-time algorithm for the 
  problem exists, but we do not know whether such an algorithm has implications 
  for the strong exponential time hypothesis.
However, we know by~\cite{DHMTW14} that the term $O(n)$ in the exponent is 
asymptotically tight, in the sense that a $2^{o(n)}$-time algorithm for 
counting perfect matchings would violate the (randomized) exponential time 
hypothesis (ETH) by Impagliazzo and Paturi~\cite{IP01}.
Using the idea of block interpolation, Curticapean~\cite{curticapean2015block} 
strengthened the hardness by showing that a $2^{o(n)}$-time algorithm for 
counting perfect matchings would violate the (deterministic) counting 
exponential time hypothesis (\#ETH).

Our main results are hardness results under \#ETH for 1) the problem of counting 
all forests in a graph, that is, its acyclic subgraphs, and 2) the problem of 
counting the number of independent sets in a bipartite graph.
If \#ETH holds, then neither of these problems has an algorithm running in time 
$\exp(o(n))$ even in simple $n$-vertex graphs of bounded maximum degree.
We use these results to lift two known ``FP vs.~\#P-hard'' dichotomy theorems to 
their more refined and asymptotically tight ``FP vs.~\#ETH-hard'' variants.

\subsection{The Tutte polynomial under \#ETH}
The Tutte polynomial of a graph~$G$ with $G=(V,E)$ is the bivariate polynomial 
$T(G;x,y)$ defined via
\begin{align}
  T(G;x,y)
  =
  \sum_{A\subseteq E}
  \paren{x-1}^{k(A)-k(E)}
  \paren{y-1}^{k(A)+|A|-|V|}
  \,,
\end{align}
where $k(A)$ is the number of connected components of the graph $(V,A)$.
The Tutte polynomial captures many combinatorial properties of a graph in a 
common framework, such as the number of spanning trees, forests, proper 
colorings, and certain flows and orientations, but also less obvious connections 
to other fields, such as link polynomials from knot theory, reliability 
polynomials from network theory, and (perhaps most importantly) the Ising and 
Potts models from statistical physics.
We make no attempt to survey the literature or the different applications for 
the Tutte polynomial, and instead refer to the upcoming CRC handbook on the 
Tutte polynomial \cite{CRChandbookTutte}.

Since $T(G;-2,0)$ corresponds to the number of proper $3$-colorings of~$G$, we 
cannot hope to compute all coefficients of $T(G;x,y)$ in polynomial time.
Instead, the literature and this paper focus on the complexity of evaluating the 
Tutte polynomial at fixed evaluation points.
That is, for each $(x,y)\in\Q^2$, we consider the function $T_{x,y}$ defined as 
$G\mapsto T(G;x,y)$.
Jaeger, Vertigan, and Welsh~\cite{JVW90} proved that this function is either 
\cc{\#P}-hard to compute or has a polynomial-time algorithm.
In particular, if $(x,y)$ satisfies $(x-1)(y-1)=1$, then $T_{x,y}$ corresponds 
to the $1$-state Potts model and has a polynomial-time algorithm, and if $(x,y)$ 
is one of the four points $(1,1)$, $(-1,-1)$, $(0,-1)$, or $(-1,0)$, it also has 
a polynomial-time algorithm; the most interesting point here is $T(G;1,1)$, 
which corresponds to the number of spanning trees in~$G$ and happens to admit a 
polynomial-time algorithm.

A trivial algorithm to compute the Tutte polynomial runs in time $2^{O(m)}$, 
where $m$ is the number of edges.
Björklund et al.~\cite{BHKK08} proved that there is an algorithm running in time 
$\exp\paren[\big]{O(n)}$, where $n$ is the number of vertices.
Dell et al.~\cite{DHMTW14} proved for all hard points, except for points with 
$y=1$, that an $\exp\paren[\big]{o(n/\log^3 n)}$-time algorithm for~$T_{x,y}$ on 
simple graphs would violate \#ETH.
Distressingly, this result not only left open one line, but also left a gap in 
the running time.
Curticapean~\cite{curticapean2015block} introduced the technique of block 
interpolation to close the running time gap: Under \#ETH, there does not exist 
an $\exp\paren[\big]{o(n)}$-time algorithm for~$T_{x,y}$ on simple graphs at any 
hard point $(x,y)$ with $y\ne 1$.\footnote{The conference version 
  of~\cite{curticapean2015block} does not handle the case $y=0$, but the full 
  version (to appear) does}

\medskip\noindent{\textbf{Our contributions:} }
We resolve the complexity of the missing line $y=1$ under \#ETH.
On this line, the Tutte polynomial counts forests weighted in some way, and the 
main result is the following theorem.
\begin{theorem}[Forest counting is hard under \#ETH]
  \label{thm:forest counting}
  If \#ETH holds, then there exist constants $\epsilon,C>0$ such that 
  no $O(\exp(\epsilon n))$-time algorithm can compute the number of all forests 
  in a given simple $n$-vertex graph with at most $C\cdot n$ edges.
\end{theorem}

The fact that the problem remains hard even on simple sparse graphs
makes the theorem stronger.
The previously best known lower bound under \#ETH was that forests cannot be 
counted in time $O\paren*{\exp(n^\delta)}$ where $\delta>0$ is some constant 
depending on the instance blow-up caused by the known \#P-hardness reductions 
for forest counting; our rough estimate suggests that $\delta$ is much smaller 
than $1/2$.
Our approach yields a \#P-hardness proof for forest counting that is simpler 
than the proofs we found in the literature.\footnote{For the \#P-hardness of 
  forest counting, \cite{JVW90} refers to private communication with Mark Jerrum 
  as well as the PhD thesis of Vertigan~\cite{vertigan1991computational}.  A 
  self-contained (but involved) proof appears in ``Complexity of Graph 
  Polynomials'' by Steven D. Noble, chapter~13 
  of~\cite{ginestet2008combinatorics}.}

Combined with all previous results \cite{JVW90,DHMTW14,curticapean2015block}, we 
can now formally state a complete \#ETH dichotomy theorem for the Tutte 
polynomial over the reals.
\begin{corollary}[Dichotomy for the real Tutte plane under \#ETH]
  \mbox{}\\
  Let $(x,y)\in\Q^2$.
  If $(x,y)$ satisfies
  \begin{center}
    $(x-1)(y-1)=1$ or $(x,y)\in\set[\big]{(1,1),(-1,-1),(0,-1),(-1,0)}$,
  \end{center}
  then $T_{x,y}$ can be computed in polynomial time.
  Otherwise $T_{x,y}$ is \#P-hard and, if \#ETH is true, then there exists 
  $\epsilon>0$ %and $C\in\N$
  such that $T_{x,y}$ cannot be computed in time $\exp(\epsilon n)$,
  even for simple graphs with $n$ vertices
\end{corollary}
The result also holds for sparse simple graphs.
We stated the results only for rational numbers in order to avoid insightless 
discussions about how real numbers should be represented.

For the proof of Theorem~\ref{thm:forest counting}, we establish a reduction 
chain that starts with the problem of counting perfect matchings on sparse 
graphs, which is known to be hard under \#ETH.
As an intermediate step, we find it convenient to work with the multivariate 
forest polynomial as defined, for example, by Sokal~\cite{Sokal2005}.
After a simple transformation of the graph, we are able extract the number of 
perfect matchings of the original graph from the multivariate forest polynomial 
of the transformed graph, even when only two different variables are used.
Subsequently, we use Curticapean's idea of block 
interpolation~\cite{curticapean2015block} to reduce the problem of computing all 
coefficients of the bivariate forest polynomial to the problem of evaluating the 
univariate forest polynomial on multigraphs where all edge multiplicity are 
bounded by a constant.
Finally, we replace parallel edges with parallel paths of constant length to 
reduce to the problem of evaluating the univariate forest polynomial on simple 
graphs.
 
\subsection{\#CSP over the Boolean domain under \#ETH}

In the second part of this paper, we consider constraint satisfaction problems 
(CSPs) over the Boolean domain, which are a natural generalization of the 
satisfiability problem for $k$-CNF formulas.
A~\emph{constraint} is a relation $R\subseteq\zo^k$ for some~$k\in\N$, and a 
set~$\Gamma$ of constraints is a \emph{constraint language}.
CSP$(\Gamma)$ is the constraint satisfaction problem where all constraints 
occurring in the instances are of a type contained in~$\Gamma$, and 
\#CSP$(\Gamma)$ is the corresponding counting version, which wants to compute 
the number of satisfying assignments.
If all constraints happen to be \emph{affine}, that is, they are linear 
equations over $\GF{2}$, then the number of solutions can be determined in 
polynomial time by applying Gaussian elimination and determining the size of the 
solution space.
Creignou and Hermann~\cite{creignou1996complexity} prove that, as soon as you 
allow even just one constraint type that is not affine, the resulting CSP 
problem is \#P-hard.

\medskip\noindent{\textbf{Our contributions:} }
We prove that all Boolean \#CSPs that are \#P-hard are also hard under \#ETH.
The \#P-hardness is established in~\cite{creignou1996complexity} by reductions 
from counting independent sets in bipartite graphs, which we prove to be hard in 
the following theorem.

\begin{theorem}[Counting independent sets in bipartite graphs is hard under \#ETH]
  \label{thm:bishard}
  If \#ETH holds, then there exist constants $\epsilon>0$ and $D\in\N$ such that 
  no $O(\exp(\epsilon n))$-time algorithm can compute the number of independent 
  sets in bipartite $n$-vertex graphs of maximum degree at most~$D$.
\end{theorem}

The fact that the problem is hard even on graphs of bounded degree makes the 
theorem stronger.
We remark that the number of independent sets in bipartite graphs has a 
prominent role in counting complexity.
Currently, the complexity of \emph{approximating} this number is unknown, and 
many problems in approximate counting turn out to be polynomial-time equivalent 
to approximately counting independent sets in bipartite graphs. 
Theorem~\ref{thm:bishard} shows that this mysterious situation does not occur 
for the exact counting problem in the exponential-time setting: it is just as 
hard as counting satisfying assignments of $3$-CNFs.

The $\#\P$-hardness of counting independent sets in bipartite graphs was 
established by Provan and Ball \cite{provan1983complexity}.
The main ingredient in their proof is a system of linear equations that does not 
seem to correspond to polynomial interpolation directly.
We prove Theorem~\ref{thm:bishard} by transferring the block interpolation idea 
from~\cite{curticapean2015block} to this system of linear equations, which we do 
using a Kronecker power of the original system.

Theorem~\ref{thm:bishard} combined with existing reductions 
in~\cite{creignou1996complexity} yields the fine-grained dichotomy.
\begin{corollary}[Creignou and Hermann under \#ETH] \label{thm:candh}
  Let $\Gamma$ be a finite constraint language.
  If every constraint in~$\Gamma$ is affine, then $\#\mathrm{CSP}(\Gamma)$ has a 
  polynomial-time algorithm.
  Otherwise $\#\mathrm{CSP}(\Gamma)$ is $\#\P$-complete, and if \#ETH holds, it 
  cannot be computed in time $\exp(o(n))$ where $n$ is the number of variables.
\end{corollary}
We consider Corollary~\ref{thm:candh} to be a first step towards understanding 
the fine-grained complexity of technically much more challenging dichotomies, 
such as the ones for counting CSPs with complex weights of Cai and 
Chen~\cite{cai2012complexity}, or the dichotomy for Holant problems with 
symmetric signatures over the Boolean domain of Cai, Lu and 
Cia~\cite{cai2011computational}.  

\section{Preliminaries}

Given a matrix $A$ of size $m_1 \times n_1$ and a matrix $B$ of size $m_2 \times n_2$ their \emph{Kronecker product} $A \otimes B$ is a matrix of size $m_1 m_2 \times n_1n_2$ given by
\[A \otimes B =
\begin{bmatrix}
	a_{11}B & \hdots & a_{1n}B\\
	\vdots & \vdots & \vdots\\
	a_{m1}B & \hdots & a_{mn}B
\end{bmatrix}
\,.
\]

Let $A^{\otimes n}$ be the matrix defined by
$A^{\otimes 1} = A$ and $A^{\otimes n+1} = A \otimes A^{\otimes n}$.
Furthermore, if~$A$ and~$B$ are quadratic matrices of size $n_a$ and $n_b$, 
respectively, then $\mathrm{det}(A \otimes B) = \mathrm{det}(A)^{n_b} \cdot 
\mathrm{det}(B)^{n_a}$

The exponential time hypothesis (ETH) by Impagliazzo and Paturi~\cite{IP01} is 
that satisfiability of $3$-CNF formulas cannot be computed substantially faster 
than by trying all possible assignments.
The counting version of this hypothesis~\cite{DHMTW14} reads as follows:

\vspace{12pt}
\begin{minipage}{.95\linewidth}
  \hspace{\parindent}
  (\#ETH)
  \quad
  \begin{minipage}{.85\linewidth} \textit{
    There is a constant $c>0$ such that no deterministic algorithm
    can compute \pp{\#$3$-SAT} in time~$\exp(c\cdot n)$, where $n$ is the number 
    of variables.}
  \end{minipage}
\end{minipage}
\vspace{12pt}

A different way of formulating \#ETH is to say no algorithm can compute 
\pp{\#$3$-SAT} in time $\exp(o(n))$.
The latter statement is clearly implied by the formal statement, and it will be 
more convenient for discussion to use this form.

The sparsification lemma by Impagliazzo, Paturi, and Zane~\cite{ImpagliazzoPZ01} 
is that every $k$-CNF formula $\varphi$ can be written as the disjunction of 
$2^{\epsilon n}$ formulas in $k$-CNF, each of which has at most 
$c(k,\epsilon)\cdot n$ clauses.
Moreover, this disjunction of sparse formulas can be computed from~$\varphi$ 
and~$\epsilon$ in time $2^{\epsilon n}\cdot\poly(m)$.
The density $c=c(k,\epsilon)$ is the \emph{sparsification constant}, and the 
best known bound is $c(k,\epsilon)=(k/\epsilon)^{3k}$~\cite{CalabroIP06}.
It was observed~\cite{DHMTW14} that the disjunction can be made so that every 
assignment satisfies at most one of the sparse formulas in the disjunction, and 
so the sparsification lemma applies to \#ETH as well.
In particular, \#ETH implies that \pp{\#$3$-SAT} cannot be computed in time 
$\exp(o(m))$, where $m$ is the number of clauses.

We also make use of the following result, whose proof is based on block 
interpolation.
\begin{theorem}[Curticapean~\cite{curticapean2015block}]\label{thm:curt}
  If \#ETH holds, then there are constants $\epsilon,D > 0$ such that neither of 
  the following problems have $O(2^{\epsilon n})$-time algorithms, even for 
  simple $n$-vertex graphs~$G$ of maximum degree at most~$D$:
  \begin{itemize}
    \item Computing the number of perfect matchings of $G$.
    \item Computing the number of independent sets of $G$.
  \end{itemize}
\end{theorem}

\section{Counting forests is \#ETH-hard}
Let $\forests(G)$ be the set of all \emph{forests} of $G$, that is, edge subsets 
$A\subseteq E(G)$ such that the graph $(V(G),A)$ is acyclic.
For $y=1$, only the terms with $k(A)+\abs{A}-\abs{V}=0$ survive, and we get the 
following:
\begin{align*}
  T(G;x,1)
  =
  \sum_{\substack{A\in\forests(G)}}
  \paren{x-1}^{k(A)-k(E)}
  \,.
\end{align*}
We want to prove that, for every fixed $x\ne 1$, computing the value $T(G;x,1)$ 
for a given graph~$G$ is hard under \cETH{}.
In particular, this is true for $T(G;2,1)$, which is the number of forests 
in~$G$.
The goal of this section is to prove the following theorem.
\begin{theorem}\label{thm:forest counting line}
  Let $x\in\R\setminus\set{1}$.
  If \cETH{} holds, then there exist $\epsilon,C>0$ such that the 
  function that maps simple $n$-vertex graphs~$G$ with at most 
  $C\cdot n$ edges to the value $T(G;x,1)$ 
  cannot be computed in time $2^{\epsilon n}$.
\end{theorem}
Theorem~\ref{thm:forest counting line} yields Theorem~\ref{thm:forest counting} 
as its special case with $x=2$.

\subsection{The multivariate forest polynomial}
A \emph{weighted graph} is a graph~$G$ in which every edge~$e\in E(G)$ is 
endowed with a weight $w_e$, which is an element of some ring.
We use the \emph{multivariate forest polynomial}, defined e.g.~by
Sokal~\cite[(2.14)]{Sokal2005} as follows:
 \begin{align*}
  F(G;w)
  &=
  \sum_{\substack{A\in\forests(G)}}
  \prod_{e\in A}
  w_e
  \,.
\end{align*}
Projecting all weights $w_e$ onto a single variable $x$ yields the 
\emph{univariate forest polynomial}:
 \[
 	F(G;x) = \sum_{\substack{A\in\forests(G)}} x^{|A|} = \sum_{k=0}^{|E(G)|} a_k(G) x^k
 \,,
 \]
where $a_k(G)$ is the number of forests with $k$ edges in $G$.
For all $x\in\R\setminus\set{1}$, the formal relation between $T(G;x,1)$ and the 
univariate forest polynomial is given by the identity
 \begin{align} \label{eq:tutte_and_forest}
 T(G;x,1) = (x-1)^{|V|- k(E)} \sum_{A \in \forests(G)} (x-1)^{-|A|}
= (x-1)^{|V|- k(E)} \cdot F\paren[\Big]{G;\frac{1}{x-1}}
\,.
 \end{align}
The first equality follows from the fact that $A$ is a forest if and only if 
$k(A)+|A|-|V|=0$ holds.
As a result, evaluating the forest polynomial and evaluating the Tutte 
polynomial for $y=1$ are problems that are polynomial-time equivalent.
 
For a forest $A\in\forests(G)$, let $c(A)$ be the family of all sets $T\subseteq 
V(G)$ such that~$T\ne\emptyset$ and $T$ is a maximal connected component in~$A$;  
clearly, each such~$T$ is the vertex set of a tree in the forest, where we also 
allow trees with $\abs{T}=1$.
\begin{lemma}[Adding an apex]\label{lem: apex}
  Let $G$ be a weighted graph, and let $G'$ be obtained from~$G$ by adding a new 
  vertex~$a$ and joining it with each vertex~$v\in V(G)$ using an edge of 
  weight~$z_v$.
  Then
  \begin{align}\label{eq:adding an apex}
    F(G')
    =
    \sum_{A\in\forests(G)}
    \prod_{e\in A} w_e
    \cdot
    \prod_{T\in c(A)}
    \paren*{1 + \sum_{v\in T} z_v}
    \,.
  \end{align}

  Moreover, when we set $z_v=-1$ for all $v\in V(G)$ and $w_e=w$ for all $e\in 
  E(G)$, we have that the coefficient of $w^{n/2}$ in $F(G')$ is equal to the 
  number of perfect matchings in $G$.
\end{lemma}
\begin{proof}
  In order to prove the claim, we first define a projection~$\phi$ that maps 
  forests $A'$ in the graph~$G'$ to forests $A=\phi(A')$ in the original 
  graph~$G$.
  In particular, $\phi$ simply removes all edges added in the construction 
  of~$G'$, that is, we define $\phi(A')=E(G)\cap A'$ for all 
  $A'\in\forests(G')$.
  Now $\phi(A')$ is a forest in $G$.

  Next we conveniently characterize the forests~$A'$ that map to the same~$A$ 
  under~$\phi$.
  Let $A$ be a fixed forest in~$G$.
  Then a forest~$A'$ in~$G'$ maps to $A$ under $\phi$ if and only if there is a 
  set $X$ with $X=A'\setminus A$ such that the following property holds:
  \begin{description}
    \item{(P)}
      For all trees $T\in c(A)$, at most one edge of~$X$ is incident to a vertex 
      of~$T$.
  \end{description}
  The forward direction of this claim follows from the fact that~$A'$ is a 
  forest, and so in addition to any tree~$T\in c(A)$ it can contain at most one 
  edge connecting~$T$ to~$a$; otherwise the tree and the two edges to~$a$ would 
  contain a cycle in~$A'$.
  For the backward direction of the claim, observe that adding a set $X$ with 
  the property (P) to~$A$ cannot introduce a cycle.

  Finally, we calculate the weight contribution of all~$A'$ that map to the 
  same~$A$.
  Let $A'$ be a forest in~$G$, let $A=\phi(A')$ and $X=A'\setminus A$.
  The weight contribution of~$A'$ in the definition of~$F(G')$ is $\prod_{e\in 
    A'} w'_e$.
  For all $e\in A$, we have $w'_e=w_e$.
  For all $e\in X$, we let $v_e\in V(G)$ be the vertex with $e=\set{a,v_e}$, and 
  we have $w'_e=z_{v_e}$.
  Thus the overall weight contribution of all $A'$ with $\phi(A)=A'$ is
  \begin{align}\label{eq: contribution for A}
    \sum_{\substack{A'\in\forests(G')\\\phi(A')=A}} \prod_{e\in A} w'_e
    =
    \prod_{e\in A} w_e \cdot \sum_{X} \prod_{e\in X} z_{v_e}
    =
    \prod_{e\in A} w_e \cdot \prod_{T\in c(A)} \paren*{1+\sum_{v\in T} z_{v}}
    \,.
  \end{align}
  The sum in the middle is over all~$X$ with the property (P), and the first 
  equality follows from the bijection between forests~$A'$ and sets~$X$ with 
  property~(P).
  For the second equality, we use property (P) and the distributive law.
  We obtain \eqref{eq:adding an apex} by taking the sum of equations~\eqref{eq: 
    contribution for A} over all~$A\in\forests(G)$.

  For the moreover part of the lemma, note that the stated settings of the edge 
  weights for~$G'$ yields
  \begin{align*}
    F(G')=\sum_{A\in\forests(G)} w^{\abs{A}} \prod_{T\in c(A)} 
    \paren*{1-\abs{T}}
    \,.
  \end{align*}
  The coefficient of $w^{n/2}$ in $F(G')$ satisfies
  \begin{align}\label{eq: extracted nhalf}
    [w^{n/2}] F(G')=\sum_{\substack{A\in\forests(G)\\\abs{A}=n/2}} \prod_{T\in 
      c(A)} \paren*{1-\abs{T}}
    \,.
  \end{align}
  Since $(V(G),A)$ is an acyclic graph with exactly~$n/2$ edges, it is either a 
  perfect matching or it contains an isolated vertex.
  If it contains an isolated vertex~$v$, then we have $\set{v}\in c(A)$ and thus 
  the product in~\eqref{eq: extracted nhalf} is equal to zero.
  It follows that~$A$ does not contribute to the sum if it is not a perfect 
  matching.
  On the other hand, if~$A$ is a perfect matching, then we have $\abs{T}=2$ for 
  all $T\in c(A)$, so the product in~\eqref{eq: extracted nhalf} is equal 
  to~$1$.
  Overall, we obtain that $[w^{n/2}] F(G')$ is equal to the number of perfect 
  matchings of~$G$.
\end{proof}

Lemma~\ref{lem: apex} shows that computing the multivariate forest polynomial is 
at least as hard as counting perfect matchings; moreover, this is true even if 
at most two different edge weights are used.
Next we argue how to reduce from the multivariate forest polynomial with at most 
two distinct weights to the problem of evaluating the univariate polynomial in 
multigraphs.
We do so via an oracle serf-reduction, whose queries are sparse multigraphs in 
which each edge has at most a constant number of parallel edges.
\begin{lemma}[From two weights to small weights using block interpolation]
\label{bivariate_to_const}
  Let $x$ and $y$ be two variables, and let $z \in \R\setminus\set{0}$ be fixed.
  There is an algorithm as follows:
  \begin{enumerate}
    \item Its input is a weighted graph~$(G,w)$ with $w_e\in\set{x,y}$ for all 
      $e\in E(G)$, and a real $\epsilon>0$.
    \item It outputs all coefficients of the bivariate polynomial $F(G;w)$.
    \item It runs in time $2^{\epsilon \abs{E(G)}} \cdot \poly(\abs{G})$.
    \item It has access to an oracle that computes $F(G;z\cdot w')$, 
          where $w'$ can be any weight function that assigns integer weights~$w'_e$
          satisfying $0\le w'_e \le C_\epsilon$ for some constant $C_\epsilon$
          that only depends on~$\epsilon$.
  \end{enumerate}
\end{lemma}
We remark that non-negative integer multiples of $z$, say~$z\cdot w'_e$,
can be thought of as $w'_e$ parallel edges of weight $z$ in a multigraph.
The quantity $F(G';z)$ for this multigraph $G'$ is then equal to the value 
$F(G;z\cdot w')$ of the weighted forest polynomial of~$G$ at $z\cdot w'$.
In particular, $F(G';1)$ is the number of forests in $G'$.
\begin{proof}
  Let $(G,w)$ with $w\in\set{x,y}^E$ and $\epsilon>0$ be given as input.
  We define $C_\epsilon\in\N$ as a large enough constant to be determined 
  later.
  The algorithm now assigns a new weight $z \cdot w'_e$ to each edge~$e$,
  where each $w'_e$ is chosen from the set of indeterminates
  $X \cup Y$ with $ X= \set{x_1,\ldots,x_{m/C}}$ and 
  $Y=\set{y_1,\ldots,y_{m/C}}$ in the following way:
  If $w_e = x$, choose $w'_e \in X$, and if $w_e = y$, choose $w'_e \in Y$.
  We further demand that the number of edges sharing the same weight
  is at most $C$ for each weight in $X \cup Y$.
  Among all such assignments~$z\cdot w'$, we pick an arbitrary one.
  We now consider the polynomial $F(G;z\cdot w')$.
  It has at most $2m/C$ variables and the maximum degree of each variable
  is at most~$C$, so $F(G;z \cdot w')$ has at most $(C+1)^{2m/C}$ monomials, where 
  $m=\abs{E(G)}$.
  The coefficients of this polynomial can be reconstructed when its values are 
  given for all evaluation points in the grid $(z\cdot[0,C])^{2m/C}$.

  Since each evaluation point only uses non-negative integer multiples of $z$
  between~$0$ and~$z\cdot C$, we can obtain the values at these evaluation 
  points by querying the oracle for $F(G;z\cdot w')$ that we are given.
  The number of evaluation points in the grid is equal to $\paren*{C+1}^{2m/C}$.
  The claim on the running time follows since the interpolation can be performed 
  in time $\poly\paren[\Big]{\paren*{C+1}^{2m/C}}$, which is at most $C\cdot 
  2^{\epsilon m}$ when~$C$ is chosen large enough depending on~$\epsilon$.

  In order to obtain the coefficient of $x^iy^j$ in $F(G;w)$, we compute the 
  image of $F(G;z\cdot w')$ under the projection that maps all variables in $X$ 
  to $x/z$ and all variables in $Y$ to $y/z$.  That is, we sum up the 
  coefficients of $F(G;z\cdot w')$ corresponding to the same monomial~$x^iy^j$,
  and divide by the factor $z^{i+j}$.
\end{proof}

The combination of Lemma~\ref{lem: apex} and Lemma~\ref{bivariate_to_const} 
shows, for all fixed $x \neq 0$, that it is hard to evaluate $F(G;x)$ for 
multigraphs with at most a constant number of parallel edges.
Next we apply a stretch to make the graphs simple.
To this end, we calculate the effect of a $k$-stretch on the univariate forest 
polynomial of a graph.

\begin{lemma}[The forest polynomial under a $k$-stretch]
\label{lem:kstretch}
  Let $G$ be a multigraph with $m$ edges, where every edge is weighted 
  with~$w\in\R$ and let $k$ be a positive integer such that the number $g_k(w)$ 
  with
  \[
    g_k(w)=\frac{w^k}{(w+1)^k - w^k}
  \]
  is well-defined.
  Let $G'$ be the simple graph obtained from $G$ by replacing every edge by a 
  path of $k$ edges.
  Then we have
  \begin{align*}
    F(G';w)
    =
    \paren[\big]{(w+1)^k - w^k}^m \cdot F\paren[\big]{G;g_k(w)}
    \,.
  \end{align*}
\end{lemma}
\begin{proof}
  We define a mapping $\phi$ that maps forests in $G'$ to forests in $G$ as 
  follows:
  We add an edge~$e\in E(G)$ to $A=\phi(A')$ if and only if $A'$ contains all 
  $k$ edges of~$G'$ that~$e$ got stretched into.
  That is, subgraphs~$A'$ that only differ by edges in ``incomplete paths'' are 
  mapped to the same multigraph~$A$ by $\phi$.

  Clearly, $\phi$ partitions $\forests(G')$ into sets of forests with the same 
  image under $\phi$.
  Let $A$ be a forest in $G$, and let us describe a way to generate all $A'$ 
  with $\phi(A')=A$.
  First, for each $e\in A$, we add its corresponding path in $G'$ of length~$k$ 
  to $A'$.
  Moreover, for each edge $e \in E(G)\setminus A$, we can add to~$A'$ any proper 
  subset of edges from the $k$-path in $G'$ that corresponds to $e$.
  Therefore, at each $e\in E(G)\setminus A$ independently, there are 
  $\binom{k}{i}$ ways to extend~$A'$ by $i$ edges to a forest in~$G'$.
  A forest $A'$ can be obtained in this fashion if and only if $\phi(A')=A$ 
  holds.

  For a fixed~$A$, let us consider all summands~$w^{\abs{A'}}$ in $F(G';w)$ with 
  $\phi(A')=A$.
  By the above considerations, the total weight contribution of these summands 
  is $w^{k\cdot\abs{A}}\cdot\paren[\big]{\sum_{i=0}^{k-1} 
    \binom{k}{i}w^i}^{m-\abs{A}}$,
  which equals
  $w^{k\cdot\abs{A}}\cdot\paren[\big]{\paren{w+1}^k-w^k}^{m-\abs{A}}$ by the 
  binomial theorem.
  These remarks justify the following calculation for the forest polynomial:
  \begin{align*}
    F(G';w)
    &=
    \sum_{A\in\forests(G)}
    \sum_{\substack{A'\in\forests(G')\\\phi(A')=A}}
    w^{\abs{A'}}
    =
    \sum_{A\in\forests(G)}
    w^{k\cdot\abs{A}}
    \cdot
    \paren[\Big]{(w+1)^k - w^k}^{m-\abs{A}}\\
    &=
    \paren[\Big]{(w+1)^k - w^k}^{m}
    \cdot
    \sum_{A\in\forests(G)}
    \paren*{\frac{w^k}{(w+1)^k - w^k}}^{\abs{A}}
    \,.
  \end{align*}
  Since the sum in the last line is equal to $F\paren[\big]{G;g_k(w)}$, this 
  concludes the proof.
\end{proof}

We are now in position to formally prove the main theorem of this section.
\begin{proof}[of Theorem~\ref{thm:forest counting line}]
Let $x\in\R\setminus\set{1}$.
Suppose that, for all~$\epsilon>0$, there exists an algorithm~$B$ to compute the 
mapping $G \mapsto T(G;x,1)$ in time $2^{\epsilon n}$ for given simple 
graphs~$G$ with at most $C'_\epsilon n$ edges, where $C'_\epsilon$ will be 
chosen later.
By~\eqref{eq:tutte_and_forest}, algorithm~$B$ can be used to compute 
values~$F(G;\frac{1}{x-1})$ with no relevant overhead in the running time, so 
let $t=\frac{1}{x-1}$.
Given such an algorithm (or family of algorithms), we devise a similar algorithm 
for counting perfect matchings, which together with Theorem~\ref{thm:curt} 
implies that \#ETH is false.

Let $G$ be a simple $n$-vertex graph with at most $C \cdot n$ edges.
Let $G'$ be the graph obtained from $G$ as in Lemma~\ref{lem: apex} by adding an 
apex, labeling the edges incident to the apex with the indeterminate $z$,
and all other edges with the indeterminate $w$.
By Lemma~\ref{lem: apex}, the coefficients of the corresponding bivariate forest 
polynomial of~$G'$ are sufficient to extract the number of perfect matchings 
of~$G$, so it remains to compute these coefficients.

To obtain the coefficients, we use Lemma~\ref{bivariate_to_const}.
The reduction guaranteed by the lemma produces $2^{\epsilon m}$ multigraphs~$H$, 
all with the same vertex set $V(G')$.
Moreover, each~$H$ has at most $C_\epsilon \abs{E(G')} = 
C_\epsilon(\abs{E(G)}+n) \leq O(C_\epsilon n)$ edges, and the multiplicity of 
each edge is at most $C_\epsilon$.
Finally, each edge of each~$H$ is assigned the same weight~$z$, which we will 
choose later.

The reduction makes one query for each~$H$, where it asks for the 
value~$F(H;z)$.
Our assumed algorithm however only works for simple graphs, so we perform 
a~$3$-stretch to obtain a simple graph $H'$ with at most
$3\abs{E(H)} \leq O(C_\epsilon n)$ edges.
Lemma~\ref{lem:kstretch} allows us to efficiently compute the value~$F(H;z)$ 
when we are given the value~$F(H';t)$ and $z=g_3(t)$ holds.
Since $g_k$ is a total function whenever~$k$ is a positive odd integer, and $3$ 
is indeed odd, the value $g_3(t)$ is well-defined, and we set $z=g_3(t)$.

We set $C'_\epsilon$ large enough so that $E(H')\le C'_\epsilon\cdot n$ holds.
Trickling back the reduction chain, we can use algorithm~$B$ to compute 
$T(H';x,1)$ in time $2^{\epsilon n}$ any $\epsilon>0$.
Using~\eqref{eq:tutte_and_forest}, we get the value of $F(H';t)$ since $x\neq 
1$.
This, in turn, yields the value of~$F(H;z)$ since $(z+1)^k-z^k\ne 0$ and 
$g_3(t)=z$.
We do this for each of the $2^{\epsilon m}$ queries~$H$ that the reduction in 
Lemma~\ref{bivariate_to_const} makes.
Finally, the latter reduction outputs the coefficients of the bivariate forest 
polynomial of~$G'$, which contains the information on the number of perfect 
matchings of~$G$.

To conclude, assuming the existence of the algorithm family~$B$, we are able to 
count perfect matching in time $\poly(2^{\epsilon m})$ for all~$\epsilon>0$, 
which implies via Theorem~\ref{thm:curt} that \#ETH is false.
\end{proof}
Note that the construction from the proof of Theorem \ref{thm:forest counting} 
implies hardness of $T(G;x,1)$ for tripartite $G$,
and also in the bipartite case whenever $x \neq -1$.

%%%%%%%%%%%%%%%%%%%%%%%%%%%%%%%%%%%%%%%%%%%%%%%%%%%%%%%%%%%%%%%%%%%%%%%%%%%%%
%%%%%%%%%%%%%%%%%%%%%%%%%%%%%%%%%%%%%%%%%%%%%%%%%%%%%%%%%%%%%%%%%%%%%%%%%%%%%

\section{Counting solutions to Boolean CSPs under \#ETH}

In this section, we prove that the \cc{\#P}-hard cases of
the dichotomy theorem for Boolean CSPs by Creignou and 
Hermann~\cite{creignou1996complexity} are also hard under \cETH.
The main difficulty is to establish \cETH-hardness of counting independent sets 
in bipartite graphs.
We do so first, and afterwards observe that all other reductions 
in~\cite{creignou1996complexity} can be used without modification.

\subsection{Counting Independent Sets in Bipartite Graphs is \#ETH-hard}

We prove that the problem of counting independent sets in bipartite graphs 
admits no subexponential algorithm under \cETH, even for sparse and simple 
graphs.
\begin{proof}[of Theorem~\ref{thm:bishard}]
  We reduce from the problem of counting independent sets in graphs of bounded 
  degree; by Theorem~\ref{thm:curt}, this problem does not have a 
  subexponential-time algorithm.
  First we note that a set is an independent set if and only its complement is a 
  vertex cover.
  Hence their numbers are equal.
  We devise a subexponential-time oracle reduction family to reduce counting 
  vertex covers in general to counting them in bipartite graphs.

  Given a graph $G$ with~$n$ vertices and~$m$ edges, and a running time 
  parameter $d\in\N$, the reduction works as follows.
  We partition the edges into $\frac{|E|}{d}$ blocks of size at most~$d$ each.
  We denote the blocks by $B_1,\dots,B_{\frac{m}{d}}$.  Next, for each 
  $\vec{\ell}=(\ell_1,\dots,\ell_{\frac{m}{d}})\in\N^{m/d}$, we construct the 
  graph $G_{\vec{\ell}}$ by replacing each edge $e\in B_i$ with a copy of the 
  gadget $H_{\ell_i}$ shown in Figure~\ref{fig:provan}.
  Note that $G_{\vec{\ell}}$ is bipartite.

	\begin{figure}
		\centering
		\begin{tikzpicture}[-, scale = 0.7]
		\node[draw,circle](1) at (0,0) {$u$};
		\node[draw,circle](2) at (1.5,1.5) {$~~$};
		\node[draw,circle](3) at (1.5,-1.5) {$~~$};
		\node[draw,circle](4) at (3,1.5) {$~~$};
		\node[draw,circle](5) at (3,-1.5) {$~~$};
		\node[draw,circle](6) at (4.5,1.5) {$~~$};
		\node[draw,circle](7) at (4.5,-1.5) {$~~$};
		\node[draw,circle](8) at (6,0) {$v$};
		
		\draw(1) -- (2);
		\draw(1) -- (3);
		\draw(2) -- (4);
		\draw(3) -- (5);
		\draw(4) -- (6);
		\draw(5) -- (7);
		\draw(6) -- (8);
		\draw(7) -- (8);
		
	%	\node(9) at (3,0.15) {$\vdots$};
	
    \node(10) at (4.25,0) {$\ell$ copies};
		
		\node[circle,inner sep=1pt,fill](42) at (3,0.75) {};
		\node[circle,inner sep=1pt,fill](43) at (3,0) {};
		\node[circle,inner sep=1pt,fill](44) at (3,-0.75) {};
		\end{tikzpicture}
    \caption{\label{fig:provan}%
      The gadget $H_\ell$ of Provan and Ball~\cite{provan1983complexity} as used 
      in the proof of Theorem~\ref{thm:bishard}.
      It corresponds to an $\ell$-fattening of the edge~$\set{u,v}$, followed by 
      a $4$-stretch of each of the $\ell$ parallel edges.
    }
  \end{figure}
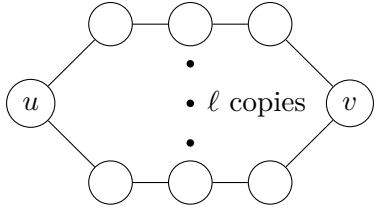
  \begin{observation}[Provan and Ball]
    \label{obs:povball}%
    The number of vertex covers of $H_{\ell}$ containing neither $u$ nor $v$ is 
    $2^{\ell}$, the number of vertex covers containing a particular one of $u$ 
    or $v$ is $3^{\ell}$, and the number of vertex covers containing both $u$ 
    and $v$ is $5^{\ell}$.
	\end{observation}
	
  We follow the proof of Provan and Ball, but do so in a block-wise fashion.
  To this end,
  let $T$ be the set of all $(m/d)\times 3$ matrices with entries from 
  $\set{0,\dots,d}$.
  The \emph{type} of a set $S\subseteq V(G)$ is the matrix $t\in T$ such that, 
  for all $i = 1 \dots \frac{m}{d}$,
  \begin{enumerate}
    \item
      $t_{i1}$ is equal to the number of edges $e\in B_i$ with $\abs{e\cap 
        S}=0$,
    \item
      $t_{i2}$ is equal to the number of edges $e\in B_i$ with $\abs{e\cap 
        S}=1$, and
    \item
      $t_{i3}$ is equal to the number of edges $e\in B_i$ with $\abs{e\cap 
        S}=2$.
  \end{enumerate}
  Every set $S\subseteq V(G)$ has exactly one type.
  Let $x_{t}$ be the number of all sets $S\subseteq V(G)$ that have type~$t$.

  We classify vertex covers $C\subseteq V(G_{\vec\ell})$ of $G_{\vec\ell}$ by 
  their intersection with $V(G)$, so let $S=C\cap V(G)$ and let $t$ be the type 
  of~$S$.
  By Observation~\ref{obs:povball} and the fact that all inserted gadgets act 
  independently after conditioning on the intersection of the vertex covers 
  of~$G_{\vec\ell}$ with~$V(G)$, there are exactly~$\prod_{i=1}^{m/\ell} 
  \paren*{2^{t_{i1}} 3^{t_{i2}} 5^{t_{i3}}}^{\ell_i}$ vertex covers~$C'$ whose 
  intersection with $V(G)$ is~$S$.
  Moreover, the number of sets~$S$ of type~$t$ is equal to~$x_t$.
  Hence the number $N_{\vec\ell}$ of vertex covers of $G_{\vec{\ell}}$ satisfies
  \begin{align}\label{eq:N235}
    N_{\vec{\ell}} = \sum_{t\in T} x_{t} \cdot \prod_{i=1}^{m/d}
    \paren*{2^{t_{i1}} 3^{t_{i2}} 5^{t_{i3}}}^{\ell_i}
  \end{align}

  Since $G_{\vec\ell}$ is bipartite, our reduction can query the oracle to 
  obtain the numbers~$N_{\vec{\ell}}$ for all $\vec{\ell} \in 
  [(d+1)^3]^{\frac{n}{d}}$.
  This yields a system of linear equations of type~\eqref{eq:N235}, where 
  the~$x_t$ for $t\in T$ are the unknowns; note that we have exactly $\abs{T}$ 
  equations and unknowns.
  Let $M$ be the corresponding $\abs{T}\times\abs{T}$ matrix, so that the system 
  can be written as $N=M\cdot x$.

  It remains to prove that $M$ is invertible.
  For this, we observe that $M$ can be decomposed into a tensor product of 
  smaller matrices as follows.
  Let $A$ be the $(d+1)^3 \times (d+1)^3$ where the row indices~$\ell$ are from 
  $[(d+1)^3]$, the column indices~$\tau$ are from~$\set{0,\dots,d}^3$, and the 
  entries are defined via
  $A_{\ell \tau} = \paren*{2^{\tau_1} 3^{\tau_2} 5^{\tau_3}}^{\ell}$.
  Provan and Ball, as well as the reader, observe that~$A$ is the transpose of a 
  Vandermonde matrix.
  Due to the uniqueness of the prime factorization, the evaluation 
  points~$2^{\tau_1} 3^{\tau_2} 5^{\tau_3}$ are distinct for distinct~$\tau$, 
  and thus $\det(A)\ne 0$.
  Furthermore, we observe that $M=A^{\otimes\frac{n}{d}}$ holds, which implies 
  $\det(M)\ne 0$ and that $M$ is invertible.

  Since $M$ is invertible, we can solve the equation system $N=M\cdot x$ in time 
  polynomial in its size, and compute $x_t$ for all $t \in T$.  Finally, we 
  compute the sum of $x_t$ over all matrices~$t$ whose first column contains 
  only zeros.
  This yields the number of all sets~$S\subseteq V(G)$ that intersect every edge 
  of~$G$ at least once, that is, the number of vertex covers of~$G$ which equals, as mentioned above, the number of independent sets of $G$.

  Assume that \#ETH holds, and let $\epsilon,D>0$ be the constants from 
  Theorem~\ref{thm:curt}, which are such that no algorithm can count independent 
  sets in general graphs of maximum degree~$D$ in time~$2^{\epsilon n}$.
  We apply our reduction to such a graph; it makes at most $2^{O(\log d \cdot 
    m/d)}$ queries to the oracle.
  Since $m\le Dn$ holds, and the running time for solving the linear equation 
  system is polynomial in the number of queries, we can choose $d\in\N$ to be a 
  large enough constant depending on~$\epsilon>0$ to achieve an overall running 
  time of~$O(2^{\frac{1}{2}\epsilon n})$ for the reduction.
  Also note that the queries to the oracle for bipartite graphs have degree at 
  most~$(d+1)^3 \cdot D$, which is a constant that only depends on~$\epsilon$.
  If there was an algorithm for counting independent sets in bipartite graphs 
  that ran in time~$O(2^{\frac{1}{2}\epsilon n})$, we would get a combined 
  algorithm for counting independent sets in general graphs that would be faster than the choice of $\epsilon$ and $D$ would allow.
  Hence, under \#ETH, there are constants $\epsilon',D'>0$ such that no 
  $O(2^{\epsilon' n})$-time algorithm can count all independent sets on graphs 
  of maximum degree at most~$D'$.
\end{proof}

We defer the simple observations needed to prove Corollary~\ref{thm:candh} to 
the appendix.

\section*{Acknowledgments}
We would like to thank Radu Curticapean for various fruitful discussions and 
interactions, and also Leslie Ann Goldberg, Miki Hermann, Mark Jerrum, John 
Lapinskas, David Richerby, and all other participants of the ``dichotomies'' 
work group at the Simons Institute in the spring of 2016.
Moreover, we would like to thank Tyson Williams who in 2013 pointed the second 
author to the self-contained \#P-hardness proof for forest counting appearing 
in~\cite{ginestet2008combinatorics}.
%, and Thore Husfeldt for encouraging us to 
%pursue the publication of this manuscript.

\bibliography{ETH-counting}

\appendix
\section{The Boolean CSP dichotomy}

Instances of the constraint satisfaction problem $\#\mathrm{CSP}(\Gamma)$ are 
conjunctions of relations in $\Gamma$ applied to variables over the Boolean 
domain and the goal is to compute the number of satisfying assignments. A 
satisfying assignment is an assignment to the variables such that the formula 
evaluates to true, that is, every relation in the conjunction evaluates to true.  
A more detailed description of the problem can be found in the paper of Creignou 
and Hermann \cite{creignou1996complexity}.

Creignou and Hermann prove Theorem~\ref{thm:candh} by reducing either from 
$\psat$, the problem of counting satisfying assignments of a $2$-CNF where every 
literal is positive, or from $\isat$, the problem of counting satisfying 
assignments of a $2$-CNF where every clause contains exactly one positive and 
one negative literal. A straightforward analysis of the construction reveals 
that the reductions only lead to a linear overhead. More precisely:

\begin{observation}
	\label{obs:candh}
  Given an instance of $\psat$ or $\isat$ with $n$ variables and a set $\Gamma$ 
  of logical relations such that at least one of the relations is not affine, 
  the construction of Creignou and Hermann results in an instance of 
  $\#\mathrm{CSP}(\Gamma)$ of size $c\cdot n$ where $c$ only depends on the size 
  of the largest non-affine relation in $\Gamma$.
\end{observation}
Therefore it suffices to establish that neither $\psat$ nor $\isat$ have a 
$2^{o(n)}$-time algorithm.
Since $\psat$ is identical to counting vertex covers in (general) graphs, 
Theorem~\ref{thm:curt} applies here.
The \cETH-hardness of $\isat$ follows by a known reduction from counting 
independent sets in bipartite graphs, which we include here for completeness.
\begin{lemma}
	\label{lem:imphard}
  Assuming \cETH, there is no algorithm that solves $\isat$ in time $2^{o(n)}$ 
  where $n$ is the number of variables.
\end{lemma} 
\begin{proof}
  Given a bipartite graph $G=(V \dotcup U, E)$ with constant degree we construct 
  a 2-CNF $F$ by adding a clause $(v \rightarrow u)$ for every edge $\{v,u\} \in 
  E$. Now the number of independent sets in $G$ equals the number of satisfying 
  assignments of $F$. Furthermore the existence of an algorithm that solves 
  $\isat$ in time $2^{o(n)}$ would imply the existence of an algorithm that 
  solves $\#\mathrm{BIS}$ in time $2^{o(n)}$. Applying Theorem~\ref{thm:bishard} 
  we obtain that such an algorithm would refute \cETH.
\end{proof}

We sketch how to obtain the \#ETH dichotomy theorem for Boolean CSPs.

\begin{proof}[of Corollary~\ref{thm:candh}]
  If every relation in $\Gamma$ is affine then we can solve 
  $\#\mathrm{CSP}(\Gamma)$ in polynomial time using Gaussian elimination as 
  in~\cite{creignou1996complexity}.
  Otherwise, the problem is \#P-hard by~\cite{creignou1996complexity}.
  If, in addition, \cETH~holds, $\#\mathrm{CSP}(\Gamma)$ cannot be solved in 
  time $2^{o(n)}$ as a subexponential algorithm could also be used to solve 
  $\psat$ or $\isat$ (see Observation~\ref{obs:candh}) in time $2^{o(n)}$ which 
  is not possible assuming \cETH{} (by Theorem~\ref{thm:curt} and 
  Lemma~\ref{lem:imphard}).
\end{proof}

\end{document}